\documentclass[a4paper,12pt,DIV11,final]{scrartcl}
\usepackage{amsmath}
\usepackage{amssymb}
\usepackage{amsthm}
\usepackage{xcolor}

\usepackage[english]{babel}
\usepackage{ifthen}
\usepackage{booktabs}
\usepackage{multirow}
\usepackage{xspace}
\usepackage{enumerate}
\usepackage{enumitem}
\usepackage{url}

\usepackage[hidelinks,final]{hyperref}

\usepackage{tikz}
\usetikzlibrary{shapes}

\definecolor{darkgreen}{rgb}{0.0,0.7,0.0}
\newenvironment{bs}{\noindent\color{darkgreen} BS:}{}

\newenvironment{mk}{\noindent\color{blue} MK:} {}

\newenvironment{vd}{\noindent\color{red} VD:} {}

\newcommand{\refthm}[1]{Theorem~\ref{#1}\xspace}

\newcommand{\reflem}[1]{Lemma~\ref{#1}\xspace}
\newcommand{\refprp}[1]{Proposition~\ref{#1}\xspace}

\newcommand{\set}[2]{\left\{#1\mathrel{\left|\vphantom{#1}\vphantom{#2}\right.}#2\right\}}
\newcommand{\oneset}[1]{\left\{\mathinner{#1}\right\}}
\newcommand{\smallset}[1]{\left\{#1\right\}}

\let\iff=\undefined

\newcommand{\iff}       {\mathrel{\Leftrightarrow}}

\newcommand{\abs}[1]{\left|\mathinner{#1}\right|}

\newcommand{\N}{\mathbb{N}}

\newcommand{\Synt}{\mathrm{Synt}}

\newcommand{\alp}{\mathrm{alph}}

\newcommand{\SF}{\mathrm{SF}}
\newcommand{\AP}{\mathrm{AP}}

\newtheorem{theorem}{Theorem}

\newtheorem{proposition}{Proposition}
\newtheorem{lemma}{Lemma}

\newtheorem{remrk}{Remark}
\newenvironment{remark}{\begin{remrk}\upshape}{\end{remrk}}

\newcommand{\ERM}{\hspace*{\fill}\ensuremath{\Diamond}} 

\setlist{itemsep=1pt,parsep=0pt,topsep=2pt}

\newcommand{\ew}{\varepsilon} 

\newcommand{\loanword}[1] {\textit{#1}}
\newcommand{\eg} {\loanword{e.g.}\xspace}
\newcommand{\ie} {\loanword{i.e.}\xspace}
\newcommand{\cf} {\loanword{cf.}\xspace}

\begin{document}

\title{Star-Free Languages \\ and Local Divisors}

\author{Manfred Kuf\-leitner \\[4.5mm]
{\normalsize FMI,
University of Stuttgart, Germany\thanks{The author gratefully acknowledges
    the support by the German Research Foundation (DFG) under grant
    \mbox{DI 435/5-1} and the support by \mbox{ANR 2010 BLAN 0202 FREC}.}} \\[-2mm] 
{\normalsize\texttt{kufleitner@fmi.uni-stuttgart.de}}}

\date{}

\maketitle

\begin{abstract}
\noindent
  \textbf{Abstract.}\,
  A celebrated result of Sch{\"u}tzenberger says
  that a language is star-free if and only if it is is recognized by a
  finite aperiodic monoid.  We give a new proof for this theorem using
  local divisors.
\end{abstract}

\section{Introduction}

The class of regular languages is built from the finite languages
using union, concatenation, and Kleene star.  Kleene showed that a language
over finite words is definable by a regular expression if and only if it is accepted by some
finite automaton~\cite{kle56}.  In particular, regular languages are
closed under complementation. It is easy to see that a language is
accepted by a finite automaton if and only if it is recognized by a
finite monoid. As an algebraic counterpart for the minimal automaton
of a language, Myhill introduced the \emph{syntactic monoid},
\cf~\cite{rs59}. 

An extended regular expression is a term over finite languages using
the operations union, concatenation, complementation, and Kleene
star. By Kleene's Theorem, a language is regular if and only if it is
definable using an extended regular expression. It is natural to ask
whether some given regular language can be defined by an extended
regular expression with at most $n$ nested iterations of the Kleene star
operation\,---\,in which case one says that the language has
generalized star height $n$. The resulting decision problem is called the
\emph{generalized star height problem}. 
Generalized star height zero means that no Kleene star operations are
allowed. Consequently, languages with generalized star height zero are called
\emph{star-free}. Sch\"utzenberger showed that a language is
star-free if and only if its syntactic monoid is
aperiodic~\cite{sch65sf}. Since aperiodicity of finite monoids is
decidable, this yields a decision procedure for generalized star height
zero. To date, it is unknown whether or not all regular languages have generalized star height one.



In this paper, we give a proof of Sch\"utzenberger's result based on \emph{local divisors}.  In commutative algebra, local divisors were introduced by Meyberg in 1972, see~\cite{FeTo02,Mey72}.  In
finite semigroup theory and formal languages, local divisors were first used by Diekert and Gastin for showing that pure future local temporal
logic is expressively complete for free partially commutative
monoids~\cite{dg06IC}.

This is a prior version of an invited contribution at the 16th International Workshop on Descriptional Complexity of Formal Systems (DCFS 2014) in Turku, Finland~\cite{kuf14dcfs}.\footnote{The final publication is available at Springer via 
\texttt{http://dx.doi.org/10.1007/}\linebreak[3]\texttt{978-}\linebreak[3]\texttt{3-}\linebreak[3]\texttt{319-}\linebreak[3]\texttt{09704-}\linebreak[3]\texttt{6\_3}.}

\section{Preliminaries}\label{sec:prem}

The set of finite words over an alphabet $A$ is $A^*$. It is the free
monoid generated by $A$. The empty word is denoted by $\ew$.  The \emph{length
  $\abs{u}$} of a word $u = a_1 \cdots a_n$ with $a_i \in A$ is $n$,
and the \emph{alphabet} $\alp(u)$ of $u$ is $\oneset{a_1, \ldots, a_n}
\subseteq A$. A language is a subset of $A^*$. The
concatenation of two languages $K,K' \subseteq A^*$ is $K \cdot K' =
\set{uv}{u \in K, v \in K'}$, and the set difference of $K$ by $K'$ is
written as~$K \setminus K'$.
Let~$A$ be a finite alphabet. The class of \emph{star-free languages}
$\SF(A^*)$ over the alphabet~$A$ is defined as follows:
\begin{itemize}
\item $A^* \in \SF(A^*)$ and $\smallset{a} \in \SF(A^*)$ for every $a
  \in A$.
\item If $K,K' \in \SF(A^*)$, then each of $K \cup K'$, $K \setminus K'$,
  and $K\cdot K'$ is in $\SF(A^*)$.
\end{itemize}
By Kleene's Theorem, a language is regular if and
only if it can be recognized by a deterministic finite
automaton~\cite{kle56}.  In particular, regular languages are closed
under complementation and thus, every star-free language is regular.


\begin{lemma}\label{lem:subsetSF}
  If $B \subseteq A$, then $\SF(B^*) \subseteq \SF(A^*)$.
\end{lemma}

\begin{proof}
  It suffices to show $B^* \in \SF(A^*)$. We have $ B^* = A^* \,
  \setminus \, \bigcup_{b \not\in B} A^* b A^*$.
\end{proof}


A monoid $M$ is \emph{aperiodic} if for every $x \in M$ there exists a
number $n \in \N$ such that $x^n = x^{n+1}$.

\enlargethispage{\baselineskip}

\begin{lemma}\label{lem:AP}
  Let $M$ be aperiodic and $x,y \in M$. Then $xy = 1$ if and only
  if $x = 1$ and $y=1$.
\end{lemma}

\begin{proof}
  If $xy = 1$, then $1 = xy = x^n y^n = x^{n+1} y^n = x \cdot 1 = x$.
\end{proof}

A monoid $M$ \emph{recognizes} a language $L \subseteq
A^*$ if there exists a homomorphism $\varphi : A^* \to M$ with
$\varphi^{-1}\big(\varphi(L)\big) = L$. A consequence of Kleene's
Theorem is that a language is regular if and only if it is
recognizable by a finite monoid, see \eg~\cite{pin86}.  The class of
\emph{aperiodic languages} $\AP(A^*)$ contains all languages $L
\subseteq A^*$ which are recognized by some finite aperiodic monoid.

The \emph{syntactic congruence $\equiv_L$} of a language $L \subseteq
A^*$ is defined as follows. For $u,v \in A^*$ we set $u \equiv_L v$
if for all $p,q \in A^*$ we have\,
  $puq \in L \iff pvq \in L$.
The \emph{syntactic monoid} $\Synt(L)$ of a language $L \subseteq A^*$ is the quotient
$A^* / \equiv_L$ consisting of the equivalence classes modulo $\equiv_L$. The \emph{syntactic homomorphism} $\mu_L
: A^* \to \Synt(L)$ with $\mu_L(u) = \set{v}{u \equiv_L v}$ satisfies
$\mu_L^{-1}\big(\mu_L(L)\big) = L$. In particular, $\Synt(L)$ recognizes $L$ and
it is the unique minimal monoid with this property, see \eg~\cite{pin86}.

Let $M$ be a
monoid and $c \in M$. We introduce a new multiplication $\circ$ on $cM
\cap Mc$. For $xc,cy \in cM \cap Mc$ we let
\begin{equation*}
  xc \circ cy = xcy.
\end{equation*}
This operation is well-defined since $x'c = xc $ and $cy' = cy$
implies $x' c y' = xc y' = x cy$. For $cx, cy \in Mc$ we have $cx
\circ cy = cxy \in Mc$. Thus, $\circ$ is associative and $c$ is the
neutral element of the monoid $M_c = (cM \cap Mc, {\circ}, c)$.  Moreover,
$M'= \set{x\in M}{cx \in Mc}$ is a submonoid of $M$ such that $M' \to
cM \cap Mc$ with $x \mapsto cx$ becomes a homomorphism. It is surjective and
hence, $M_c$ is a divisor of $(M, \cdot, 1)$ called the \emph{local divisor of $M$ at $c$}. Note that if $c^2 = c$, then $M_c$ is just the local monoid $(cMc,\cdot,c)$ at the idempotent~$c$.

\begin{lemma}\label{lem:loc}
  If $M$ is a finite aperiodic monoid and $1 \neq c \in M$, then $M_c$
  is aperiodic and $\abs{M_c} < \abs{M}$.
\end{lemma}

\begin{proof}
  If $x^n = x^{n+1}$ in $M$ for $cx \in Mc$, then $(cx)^n = cx^n =
  cx^{n+1} = (cx)^{n+1}$ where the first and the last power is in $M_c$. This shows that $M_c$ is aperiodic.
  By \reflem{lem:AP} we have $1 \not\in cM$ and thus $1 \in M \setminus M_c$.
\end{proof}

\section{Sch{\"u}tzenberger's Theorem on star-free languages}

The following proposition establishes the more difficult inclusion of Sch\"utzenberger's result $\SF(A^*) = \AP(A^*)$. Its proof relies on local divisors.

\begin{proposition}\label{prp:ap:sf}
  Let $\varphi : A^* \to M$ be a homomorphism to a finite aperiodic monoid $M$. Then for all $p \in M$ we have $\varphi^{-1}(p) \in \SF(A^*)$.
\end{proposition}


\begin{proof}
  We proceed by induction on $(\abs{M},\abs{A})$ with lexicographic order. If $\varphi(A^*) = \smallset{1}$, then depending on $p$ we either have $\varphi^{-1}(p) = \emptyset$ or $\varphi^{-1}(p) = A^*$. In any case, $\varphi^{-1}$ is in $\SF(A^*)$. Note that is includes both bases cases $M = \smallset{1}$ and $A = \emptyset$.
  Let now $\varphi(A^*) \neq \smallset{1}$. Then there exists $c \in A$ with $\varphi(c) \neq 1$. We set $B = A \setminus \smallset{c}$ and we let $\varphi_c : B^* \to M$ be the restriction of $\varphi$ to $B^*$. We have
  \begin{equation}\label{eq:1}
    \varphi^{-1}(p) \;=\;
      \varphi_c^{-1}(p) \,\cup \!\!\!\!
      \bigcup_{\scriptsize\begin{array}{c}
          p = p_1 p_2 p_3
        \end{array}} \!\!\!\!
      \varphi^{-1}_{c}(p_1) \cdot
      \big(\varphi^{-1}(p_2) \cap
      c \/ A^* \cap A^* \hspace*{-0.5pt} c \big) \cdot
      \varphi^{-1}_{c}(p_3).
  \end{equation}
  The inclusion from right to left is trivial. The other inclusion can
  be seen as follows: Every word $w$ with $\varphi(w) = p$ either does not contain the letter $c$ or we
  can factorize $w = w_1 w_2 w_3$ with $c \not\in \alp(w_1 w_3)$ and
  $w_2 \in cA^* \cap A^*c$, \ie, we factorize $w$ at the first and
  the last occurrence of $c$. Equation~\eqref{eq:1} is established by
  setting $p_i = \varphi(w_i)$. By induction on the size of the
  alphabet, we have $\varphi_c^{-1}(p_i) \in \SF(B^*)$, and thus $\varphi_c^{-1}(p_i) \in \SF(A^*)$
  by \reflem{lem:subsetSF}.

  Since $\SF(A^*)$ is closed under union and concatenation, it remains
  to show $\varphi^{-1}(p) \cap c\/A^* \cap A^* c \in \SF(A^*)$ for
  $p \in \varphi(c) M \cap M \varphi(c)$. Let
  \begin{equation*}
    T = \varphi_c(B^*).
  \end{equation*}
  The set $T$ is a submonoid of $M$. In the remainder of this
  proof, we will use~$T$ as a finite alphabet. We define a substitution
  \begin{alignat*}{2}
    \sigma : \ \,&&( B^*\hspace*{1pt} c )^* \ &\to \ T^* \\
    && v_1 c \cdots v_k c \ &\mapsto \ \varphi_c(v_1) \cdots \varphi_c(v_k)
  \end{alignat*}
  for $v_i \in B^*$. In addition, we define a homomorphism $\psi : T^* \to M_c$ with
  $M_c = (\varphi(c) M \cap M \varphi(c), \circ, \varphi(c))$ by
  \begin{alignat*}{2}
    \psi : \ && T^* &\to M_c \\
    && \varphi_c(v) &\mapsto \varphi(cvc)
  \end{alignat*}
  for $\varphi_c(v) \in T$. Consider a word $w = v_1 c \cdots v_k c$ with $k \geq 0$ and $v_i \in
  B^*$. Then
  \begin{equation}
  \begin{aligned}[b]
    \psi \bigl(\sigma(w) \bigr)
    &= \psi\bigl( \varphi_c(v_1) \varphi_c(v_2) \cdots \varphi_c(v_k) \bigr) \\
    &= \varphi(c v_1 c) \circ \varphi(c v_2 c) \circ 
    \cdots \circ \varphi(c v_k c) \\
    &= \varphi(c v_1 c v_2 \cdots c v_k c) 
    = \varphi(cw).
    \label{eq:important}
  \end{aligned}
  \end{equation}
  Thus, we have $cw \in \varphi^{-1}(p)$ if and
  only if $w \in \sigma^{-1} \bigl( \psi^{-1}(p) \bigr)$.  This shows $\varphi^{-1}(p)
  \cap c\/A^* \cap A^* c = c \cdot \sigma^{-1} \bigl( \psi^{-1}(p) \bigr)$ for every
  $p \in \varphi(c) M \cap M \varphi(c)$. In particular, it remains to
  show $\sigma^{-1} \bigl( \psi^{-1}(p) \bigr) \in \SF(A^*)$. By \reflem{lem:loc},
  the monoid~$M_c$ is aperiodic and $\abs{M_c} < \abs{M}$. Thus, by
  induction on the size of the monoid we have $\psi^{-1}(p) \in
  \SF(T^*)$, and by induction on the size of the alphabet we have
  $\varphi_c^{-1}(t) \in \SF(B^*)
  \subseteq \SF(A^*)$ for every $t \in T$. For $t \in T$ and $K,K' \in
  \SF(T^*)$ we have
  \begin{align*}
    \sigma^{-1}(T^*) &= A^* c \cup \smallset{1} \\
    \sigma^{-1}(t) &= \varphi_c^{-1}(t) \cdot c \\
    \sigma^{-1}(K \cup K')
    &= \sigma^{-1}(K) \cup \sigma^{-1}(K') \\
    \sigma^{-1}(K \setminus K') 
    &= \sigma^{-1}(K) \setminus \sigma^{-1}(K') \\
    \sigma^{-1}(K \cdot K') &= \sigma^{-1}(K) \cdot
    \sigma^{-1}(K').
  \end{align*}
  Only the last equality requires justification. The inclusion from
  right to left is trivial. For the other inclusion, suppose $w = v_1
  c \cdots v_k c \in \sigma^{-1} (K \cdot K')$ for $k \geq 0$ and $v_i
  \in B^*$. Then $\varphi_c(v_1) \cdots
  \varphi_c(v_k) \in K \cdot K'$, and thus $\varphi_c(v_1) \cdots
  \varphi_c(v_i) \in K$ and $\varphi_c(v_{i+1}) \cdots \varphi_c(v_k)
  \in K'$ for some $i \geq 0$. It follows $v_1 c \cdots v_i c \in
  \sigma^{-1}(K)$ and $v_{i+1} c \cdots v_k c \in K'$. This shows $w
  \in \sigma^{-1}(K) \cdot \sigma^{-1}(K')$.

  We conclude that $\sigma^{-1}(K) \in \SF(A^*)$ for every $K \in
  \SF(T^*)$.  In particular, we have $\sigma^{-1} \big(\psi^{-1}(p)\big) \in \SF(A^*)$.
\end{proof}

\begin{remark}
  A more algebraic viewpoint of the proof of \refprp{prp:ap:sf} is the following.
The mapping $\sigma$ can
  be seen as a length-preserving homomorphism from a submonoid of
  $A^*$\hspace*{1pt}---\hspace*{1pt}freely generated by the infinite set $B^* \hspace*{1pt}c$\hspace*{1pt}---\hspace*{1pt}onto $T^*$; and this homomorphism is defined by
  $\sigma(vc) = \varphi_c(v)$ for $vc \in B^* \hspace*{0.5pt}
  c$. The mapping $\tau : M\varphi(c) \cup \smallset{1}
  \to M_c$ with $\tau(x) = \varphi(c) \cdot x$ defines a homomorphism. Now, by Equation~\eqref{eq:important} the following diagram commutes:
  \begin{center}
    \begin{tikzpicture}[scale=0.75]
      \draw (0,0) node (Pc) {$(B^* \hspace*{0.5pt} c)^*$};
      \draw (3,0) node (Q) {$T^*$};
      \draw (0,-2) node (M) {$M\varphi(c) \cup \smallset{1}$};
      \draw (3,-2) node (Mc) {$M_c$};
      \draw[->] (Pc) -- node[above] {$\sigma$} (Q);
      \draw[->] (Pc) -- node[left] {$\varphi$} (M);
      \draw[->] (Q) -- node[right] {$\psi$} (Mc);
      \draw[->] (M) -- node[above] {$\tau$} (Mc);
    \end{tikzpicture}
  \end{center}
  
  \vspace*{-2\baselineskip}
  
  \ERM
\end{remark}

The following lemma gives the remaining inclusion of $\SF(A^*) = \AP(A^*)$. Its proof is standard; it is presented here
only to keep this paper self-contained.

\begin{lemma}\label{lem:sf:ap}
  For every language $L \in \SF(A^*)$ there exists an integer $n(L)
  \in \N$ such that for all words $p,q,u,v \in A^*$ we have
  \begin{equation*}
    p\, u^{n(L)}q \in L \ \Leftrightarrow \ p\, u^{n(L)+1}q \in L.
  \end{equation*}
\end{lemma}

\begin{proof}
  For the languages $A^*$ and $\smallset{a}$ with $a \in A$ we define $n(A^*) = 0$ and $n(\smallset{a}) = 2$.
  Let now $K,K' \in \SF(A^*)$ such that $n(K)$ and $n(K')$ exist.
  We set
  \begin{gather*}
    n(K \cup K') =  n(K \setminus K') = \max \bigl( n(K), n(K') \bigr), \\
    n(K \cdot K') = n(K) + n(K') + 1.
  \end{gather*}
  The correctness of the first two choices is straightforward. For the last equation,
  suppose $p\, u^{n(K)+n(K')+2}q \in K \cdot K'$. Then either $p\,
  u^{n(K)+1} q' \in K$ for some prefix $q'$ of $u^{n(K')+1} q$ or
  $p'\, u^{n(K')+1} q \in K'$ for some suffix $p'$ of $p u^{n(K) +
    1}$. By definition of $n(K)$ and $n(K')$ we have $p\, u^{n(K)} q' \in
  K$ or $p'\, u^{n(K')} q \in K'$, respectively. Thus $p\,
  u^{n(K)+n(K')+1}q \in K \cdot K'$. The other direction is similar:
  If $p\, u^{n(K)+n(K')+1}q \in K \cdot K'$, then $p\,
  u^{n(K)+n(K')+2}q \in K \cdot K'$. This completes the proof.
\end{proof}

\begin{theorem}[Sch{\"u}tzenberger]\label{thm:schutz}
  Let $A$ be a finite alphabet and let $L \subseteq A^*$. The following conditions are equivalent:
  \begin{enumerate}
  \item\label{aaa:schutz} $L$ is star-free.
  \item\label{bbb:schutz} The syntactic monoid of $L$ is finite and aperiodic.
  \item\label{ccc:schutz} $L$ is recognized by a finite aperiodic monoid.
  \end{enumerate}
\end{theorem}

\begin{proof}
  ``\ref{aaa:schutz}$\;\Rightarrow\;$\ref{bbb:schutz}'': Every
  language $L \in \SF(A^*)$ is regular. Thus
  $\Synt(L)$ is finite, \cf~\cite{pin86}. By \reflem{lem:sf:ap}, we see that
  $\Synt(L)$ is aperiodic. The implication ``\ref{bbb:schutz}$\;\Rightarrow\;$\ref{ccc:schutz}'' is trivial.
  If $\varphi^{-1}\big(\varphi(L)\big) = L$, then we can write $L = \bigcup_{p \in \varphi(L)} \varphi^{-1}(p)$. Therefore,
  ``\ref{ccc:schutz}$\;\Rightarrow\;$\ref{aaa:schutz}'' follows by \refprp{prp:ap:sf}.
\end{proof}

The syntactic monoid of a regular language (for instance, given by a non\-deterministic automaton) is effectively computable. 
Hence, 
from the equivalence of conditions ``\ref{aaa:schutz}'' and ``\ref{bbb:schutz}'' in \refthm{thm:schutz} it follows that star-freeness is a decidable property of regular languages. The equivalence of 
``\ref{aaa:schutz}'' and ``\ref{ccc:schutz}'' can be written as
\begin{equation*}
  \SF(A^*) = \AP(A^*).
\end{equation*}
The equivalence of ``\ref{bbb:schutz}'' and ``\ref{ccc:schutz}'' is rather trivial: The class of finite aperiodic monoids is closed under division, and the syntactic monoid of $L$ divides any monoid that recognizes $L$, see \eg~\cite{pin86}.

\subsection*{Acknowlegdements}

The author would like to thank Volker Diekert and Benjamin Steinberg for  many interesting discussions on the proof method used in~\refprp{prp:ap:sf}.


\newcommand{\Ju}{Ju}\newcommand{\Ph}{Ph}\newcommand{\Th}{Th}\newcommand{\Ch}{Ch}\newcommand{\Yu}{Yu}\newcommand{\Zh}{Zh}\newcommand{\St}{St}\newcommand{\curlybraces}[1]{\{#1\}}

\end{document}